\documentclass[conference,onecolumn]{IEEEtran}
\usepackage[cmex10]{amsmath}
\usepackage{algorithm}
\usepackage{algpseudocode}
\usepackage{amssymb}
\usepackage{amsthm}
\usepackage{bbm}
\usepackage{caption}
\usepackage{cite}
\usepackage{comment}
\usepackage{diagbox}
\usepackage{enumitem}
\usepackage{etoolbox}
\usepackage{float}
\usepackage[T1]{fontenc}
\usepackage{ifthen}
\usepackage{import}
\usepackage[utf8]{inputenc} 
\usepackage{subcaption}
\usepackage{tabularx}
\usepackage{todonotes}
\usepackage{url}

\newtheorem{theorem}{Theorem}
\newtheorem{lemma}{Lemma}
\newtheorem{definition}{Definition} 

\makeatletter
\def\th@plain{%
  \thm@notefont{}
  \normalfont
}
\makeatother
\newtheoremstyle{myremark}{}{}{\itshape}{}{\itshape}{:}{ }{}
\theoremstyle{myremark}
\newtheorem{remark}{Remark} 

\DeclareMathOperator*{\argmin}{arg\,min}

\newcommand{\N}{\mathbb{N}}

\newcommand{\R}{\mathbb{R}}

\renewcommand{\P}{\mathbb{P}}

\begin{document}

\title{Rate Distortion Approach to Joint Communication and Sensing with Markov States: Open Loop Case}

\author{%
    \IEEEauthorblockN{Colton P. Lindstrom and Matthieu R. Bloch}
    \IEEEauthorblockA{School of Electrical and Computer Engineering, Georgia Institute of Technology, Atlanta, Georgia 30332, USA\\
    Email: coltonlindstrom@gatech.edu, matthieu.bloch@ece.gatech.edu}
}

\maketitle

\begin{abstract}
We investigate a joint communication and sensing (JCAS) framework in which a transmitter concurrently transmits information to a receiver and estimates a state of interest based on noisy observations.
The state is assumed to evolve according to a known dynamical model.
Past state estimates may then be used to inform current state estimates.
We show that Bayesian filtering constitutes the optimal sensing strategy.
We analyze JCAS performance under an open loop encoding strategy with results presented in terms of the tradeoff between asymptotic communication rate and expected per-block distortion of the state.
We illustrate the general result by specializing the analysis to a beam-pointing model with mobile state tracking.
Our results shed light on the relative performance of two beam control strategies, beam-switching and multi-beam.
\end{abstract}

\section{Introduction}\label{sec:introduction}
Recently, significant attention has been directed toward research in Joint Communication and Sensing (JCAS).
This is in part because it is widely anticipated to feature in next generation communication systems \cite{liu_integrated_2022,liu_survey_2022}.
Research in JCAS spans diverse approaches, including 
full hardware implementation \cite{baquero_barneto_full-duplex_2019, liyanaarachchi_optimized_2021, wang_first_2019}, 
exploitation of current infrastructure \cite{stinco_ieee_2016, berger_signal_2010, hack_detection_2014}, 
waveform design \cite{huang_loradar_2022, shi_power_2018, liu_toward_2018}, 
security \cite{su_secure_2021, deligiannis_secrecy_2018, welling2024transmitter, welling2024low}, 
and fundamental limitations.
We focus here on the later.
Generally, sensing signals require deterministic waveforms while communication signals rely on randomness to embed information, resulting in natural tradeoffs and limitations \cite{liu_deterministic-random_2023}.

Prior works have explored these limitations and tradeoffs for various models.
Of particular relevance to the present work, \cite{zhang2011joint} has suggested a rate-distortion approach to joint transmission and state estimation in which an i.i.d. state estimation cost constraint is reinterpreted as a distortion constraint resulting in a constrained channel coding problem.
\cite{ahmadipour_information-theoretic_2022} has extended this framework to monostatic JCAS systems, and \cite{ahmadipour_collaborativeJCAS} has extended this framework to collaborative multiuser JCAS systems.
\cite{li_capacity_2023, li2024capacity} have specialized the rate-distortion analysis for JCAS systems to a directional beam pointing problem, showing how JCAS operations improve system performance during initial beam acquisition.

Whereas \cite{zhang2011joint, ahmadipour_information-theoretic_2022, ahmadipour_collaborativeJCAS, li_capacity_2023, li2024capacity} model the channel state as a sequence of i.i.d. random variables, we extend the rate-distortion JCAS framework to models in which the state sequence is not i.i.d. random.
Rather, the state evolves according to some dynamical model so that past estimates of the state inform present estimates.
Furthermore, present estimates of the state provide some predictive power of future states, which in general make single letter expressions for rate distortion regions infeasible.
The analysis assumes the state varies with each channel use.
States that evolve significantly slower may be modeled as being constant over the duration of a codeword, at which point the analysis of \cite{chang2023rate, chang_sequential_2023, joudeh2022joint, wu2024joint} is more appropriate.

We propose a general JCAS system model with mobile state and causal estimates, which we analyze using a rate distortion approach.
To illustrate potential applications, we specialize the model to a beam pointing JCAS problem and provide numerical illustrations of the associated JCAS tradeoffs.
Our results in Lemma 1 and Theorem 1 are identical to those in \cite{nikbakht2024memory}, which we discovered upon finalizing this manuscript. 
However, our subsequent analysis of open loop strategies and applications to beamforming strategies offers distinct insights.

Section \ref{sec:notation} introduces the notation used throughout the paper.
Section \ref{sec:systemModel} defines the proposed JCAS model.
Section \ref{sec:mainResults} provides general expressions for the optimal estimator and capacity distortion region.
Section \ref{sec:beamPointingExample} specializes the model to the beam pointing problem.

\section{Notation}\label{sec:notation}
The indicator (or characteristic) function of a set $\Omega$ is denoted $\mathbf{1}\{\omega\in\Omega\}$.
Sets and discrete alphabets are indicated using calligraphic letters e.g., $\mathcal{X}$.
For a discrete set $\mathcal{X}$, the set $\mathcal{P_\mathcal{X}}$ is set of all probability distributions on $\mathcal{X}$.
Depending on context, capital letters denote matrices, constants, or random variables.
$A^T$ denotes the transpose of the matrix $A$, and $\text{tr}(A)$ denotes its trace.

For $n\in\N$, a sequence of length $n$ is denoted $x^n\triangleq(x_1,x_2,\cdots,x_n)$.
For $k,j\in\N$ with $k<j$, a subsequence of length $j-k+1$ is denoted $x_k^j\triangleq(x_k,x_{k+1},\cdots,x_j)$.
The $i$th element of a sequence $x^n$ is denoted $x_i$.

The mutual information between two random variables $X$ and $Y$ is denoted $I(X;Y)$.
$\log$ denotes the natural logarithm.

\section{System Model}\label{sec:systemModel}
Consider the joint communication and sensing model shown in Figure \ref{fig:systemModel}, in which a transmitter is tasked with reliably transmitting a message $m\in\mathcal{M}=[1;M]$ over a memoryless state dependent channel $P_{YZ\vert XS}$ while simultaneously estimating the state $s_i\in\mathcal{S}$.
The encoder generates codeword symbols $x_i\in\mathcal{X}$, which are used as input to the channel.
The channel produces two outputs: measurements of the channel state fed back to the transmitter $z_i\in\mathcal{Z}$, and received codeword symbols $y_i\in\mathcal{Y}$.

The decoder is assumed to know the state sequence $s^n$ and uses the received codeword $y^n$ to form an estimated message $\hat{m}$.
The channel measurements $z^i$ are processed to form causal state estimates $\hat{s}_i\in\hat{\mathcal{S}}$ where $\hat{\mathcal{S}}\subseteq\mathcal{S}$.
A new state estimate is produced at each time step $i$, and each estimate must be causal, meaning that measurements from time $i+1$ onward cannot be used to estimate the state at time $i$.
For this work, the state is assumed to be a first order Markov chain.
As such, the distribution of the state sequence can be written as $P_{S^n}(s^n)=\prod_{i=1}^nP_{S_i\vert S_{i-1}}(s_i|s_{i-1})$ (where $s_0$ is the initial state).

The sets $\mathcal{M}$, $\mathcal{X}$, and $\mathcal{Y}$ are all assumed to be finite.
The sets $\mathcal{Z}$, $\mathcal{S}$, and $\hat{\mathcal{S}}$ are assumed to be subsets of the real numbers.

Formally, an $(M,n)$ code for the proposed JCAS model is composed of the following:
\begin{enumerate}
    \item A sequence of encoding function $f_i:\mathcal{M}\times\mathcal{Z}^{i-1}\to\mathcal{X}$, $i=1,2,\dots,n$.
    \item A sequence of state estimation functions $g_{i-1}:\mathcal{X}^{i-1}\times\mathcal{Z}^{i-1}\to\hat{\mathcal{S}}$.
    \item A decoding function $h:\mathcal{Y}^n\times\mathcal{S}^n\to\mathcal{M}$.
\end{enumerate}
It is assumed that messages are drawn uniformly from the message set.

\tikzstyle{rect} = [rectangle, rounded corners, minimum width=0.5cm, minimum height=0.5cm, text centered, align=center, draw=black]
\tikzstyle{ghost} = [circle]
\tikzstyle{arrow} = [thick, ->, >=stealth]
\usetikzlibrary{calc}
\vspace{-1em}
\begin{figure}[H]
    \centering
    \begin{tikzpicture}
        \node (encode) [rect] {Encoder\\$f_i(m,z^{i-1})$};
        \node (est) [rect, above of=encode, yshift=0.5cm] {Estimator\\ $g_{i-1}(x^{i-1},z^{i-1})$};
        \node (chan) [rect, right of=encode, xshift=1.9cm] {Channel\\$P_{YZ\vert XS}$};
        \node (state) [rect, below of=chan, yshift=-0.3cm] {$P_{S_i\vert S_{i-1}}$};
        \node (decode) [rect, right of=chan, xshift=1.4cm] {Decoder\\$h(y^n,s^n)$};
        \node (delay) [rect, right of=est, xshift=1.9cm] {Delay};

        \node (message) [ghost, left of=encode, xshift=-0.7cm] {};
        \node (estimate) [ghost, left of=est, xshift=-0.9cm] {};
        \node (receive) [ghost, right of=decode, xshift=0.2cm] {};
        
        \draw [arrow] (message.west) -- node[anchor=south] {$m$} (encode);
        \draw [arrow] (encode) -- node[anchor=north] {$x_i$} (chan);
        \draw [arrow] (encode) -| ($(est.east)+(3mm,-2mm)$) |- ([yshift=-2mm] est);
        \draw [arrow] (chan) -- node[anchor=north] {$y_i$} (decode);
        \draw [arrow] ([yshift=2mm] chan) -| ($(chan.east)+(2mm, 5mm)$) |- node[anchor=west] {$z_i$} (delay);
        \draw [arrow] ([yshift=2mm] chan) -| ($(chan.east)+(2mm, 5mm)$) |-  (delay);
        \draw [arrow] (est) -- node[anchor=south, xshift=-0.1cm] {$\hat{s}_{i-1}$} (estimate.west);
        \draw [arrow] (state) -- node[anchor=west] {$s_i$} (chan);
        \draw [arrow] (state) -| (decode.south);
        \draw [arrow] (delay.west) -| ($(encode.north)+(20mm, 3mm)$) -| (encode);
        \draw [arrow] (delay) -- node[anchor=south] {$z^{i-1}$} (est);
        \draw [arrow] (decode) -- node[anchor=south] {$\hat{m}$} (receive.east);
    \end{tikzpicture}
    \caption{Block diagram for the proposed JCAS model with dynamic state.}
    \label{fig:systemModel}
\end{figure}
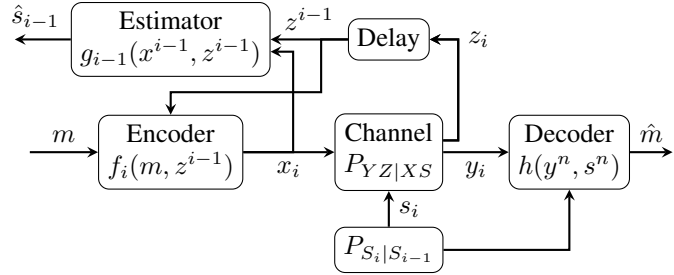
\vspace{-1em}

The performance of the model is measured in terms of asymptotic rate of reliable communication and an average per-block distortion.
The probability of communication error for a length $n$ code is defined as
\begin{equation}\label{eqn:communicationError}
    P_c^{(n)}\triangleq\max_{m\in\mathcal{M}}\mathbb{P}(h(y^n,s^n)\neq m\vert m).
\end{equation}
The communication rate is defined as $R\triangleq\frac{1}{n}\log M$.
The objective of the communication subsystem is to maximize the rate while maintaining a small probability of communication error.

The \textit{per-letter distortion function} is defined as 
\begin{equation}\label{eqn:letterDistortionDefinition}
    d_i:\mathcal{S}\times\hat{\mathcal{S}}\to\R^+ 
    \text{ denoted } d_i(s_i,\hat{s}_i)
    \text{ for } i\in[0;n].
\end{equation}
The \textit{per-block distortion function} can then be defined as the average of the per-letter distortions over $n+1$ time steps.
\begin{equation}\label{eqn:blockDistortionDefinition}
    d_{0,n}:\mathcal{S}^{n+1}\times\hat{\mathcal{S}}^{n+1}\to\R^+
    \text{, }d_{0,n}(s_0^n,\hat{s}_0^n)=\frac{1}{n+1}\sum_{i=0}^n d_i(s_i,\hat{s}_i).
\end{equation}
The objective of the state estimation subsystem is to minimize the \textit{average per-block distortion function}.
\begin{equation}\label{eqn:averageDistortionDefinition}
    \Delta^{(n)}\triangleq E\left[d_{0,n}(s_0^n,\hat{s}_0^n)\right]=\frac{1}{n+1}\sum_{i=0}^n E\left[d_i(s_i,\hat{s}_i)\right]
\end{equation}
Whether by the distortion function definition or the nature of the state space $\mathcal{S}$, the average per-block distortion is assumed upper bounded by $D_{\textnormal{max}}<\infty$.

\begin{remark}
    The per-block distortion is an average over $n+1$ time steps that includes the initial state $s_0$ and initial state estimate $\hat{s}_0$.
    It is assumed that the state estimator has an initial estimate that satisfies a user defined distortion threshold.
\end{remark}

\begin{definition}[Achievability]\label{def:achievability}
    A rate-distortion pair $(R,D)$ is achievable if, for any $\varepsilon>0$, there exists a large enough $n$ and an $(M,n)$ code such that,
    \begin{align}
        P_c^{(n)}&\leq\varepsilon\label{eqn:achievableError},\\
        \Delta^{(n)}&\leq D+\varepsilon\label{eqn:achievableDistortion},\\
        \frac{1}{n}\log M&\geq R-\varepsilon.\label{eqn:achievableRate}
    \end{align}
\end{definition}

The JCAS capacity region is the closure of the set of achievable rate-distortion pairs.

\begin{remark}
    In systems with a mobile state, scenarios may arise in which inadequate sensing leads to the complete loss of a target.
    Definition \ref{def:achievability} could be strengthened as
    \begin{equation}
        \Delta^{(i)}\leq L+\varepsilon\hspace{1em}i=0,\dots,n\label{eqn:achievableTracking}
    \end{equation}
    where $L$ is some threshold beyond which tracking of the target state is assumed lost or irretrievable. 
\end{remark}

\section{Main Results}\label{sec:mainResults}
We now characterize the open loop JCAS capacity region for the mobile state model.
This analysis is a natural generalization of \cite{ahmadipour_information-theoretic_2022}.
We begin by providing the form of the optimal estimator satisfying the constraints of the model.
This result facilitates the definition of a sensing cost and a cost-constrained input set.

\begin{lemma}\label{lem:optimalEstimator}
    Let $\hat{S}^{*n}=g^*(X^n,Z^n)$ denote the optimal state estimate sequence that minimizes the average per-block distortion $\Delta^{(n)}$.
    The optimal causal state estimate is,
    \begin{equation}\label{eqn:optimalEstimator}
        g^*(X^n,Z^n)\triangleq\left(g_1^*(X^1,Z^1),g_2^*(X^2,Z^2),\dots,g_n^*(X^n,Z^n)\right)
    \end{equation}
    where each symbol estimator $g_i^*(X^i,Z^i)$ is
    \begin{equation}\label{eqn:optimalSymbolEstimator}
        g_i^*(X^i,Z^i)=\argmin_{\hat{s}}\int_{\mathcal{S}}P_{S_i\vert X^i,Z^i}(s_i\vert x^i,z^i)d_i(s_i,\hat{s})ds_i.
    \end{equation}.
\end{lemma}
\vspace{-2em}
\begin{proof}
    See Appendix \ref{apd:optimalEstimator}
\end{proof}

\begin{remark}
    For the case where the state sequence is i.i.d., the Markov chain $(X^{i-1},Z^{i-1})-(X_i,Z_i)-S_i$ holds.
    Consequently, $P_{S_i\vert X^i,Z^i}(s_i\vert x^i,z^i)=P_{S_i\vert X_i,Z_i}(s_i\vert x_i,z_i)$ and the optimal symbol estimate becomes 
    \begin{align}
        g_i^*(X^i&,Z^i)=g_i^*(X_i,Z_i)\nonumber\\
        &=\argmin_{\hat{s}}\int_{\mathcal{S}}P_{S_i\vert X_i,Z_i}(s_i\vert x_i,z_i)d_i(s_i,\hat{s})ds_i
    \end{align}
    and the optimal block estimator of the entire sequence is determined symbol-wise, consistent with \cite{ahmadipour_information-theoretic_2022}:
    \begin{equation}
         g^*(X^n,Z^n)=\left(g_1^*(X_1,Z_1),g_2^*(X_2,Z_2),\dots,g_n^*(X_n,Z_n)\right).
    \end{equation}
\end{remark}

Lemma \ref{lem:optimalEstimator} shows that optimal estimates are generated by minimizing the expected value of a cost function $d_i(s_i,\hat{s})$ over a posterior distribution at each time step.
Given the causal, memoryless, and Markov nature of the model, the posterior distribution can be calculated using standard predict/update Bayesian estimation equations \cite{sarkka2023bayesian} with a minor variation to include the channel inputs $X^i$ in conjunction with channel measurements $Z^i$.

Let $c(x^n)\triangleq E[d_{0,n}(S^n,g^*(X^n,Z^n))\big\vert X^n]$ be the \textit{sensing cost} for a sequence of channel inputs $X^n$.
Define the set of cost constrained input sequences as
\begin{align}
    \overrightarrow{\mathcal{P}}_D^{(n)}&=
    \left\{P_{X^n}: E[d_{0,n}(S^n,g^*(X^n,Z^n))]\leq D\right\}\nonumber\\
    &=\left\{P_{X^n}: E_{X^n}\left[E[d_{0,n}(S^n,g^*(X^n,Z^n))\big\vert X^n]\right]\leq D\right\}\nonumber\\
    &=\left\{P_{X^n}: \sum_{x^n}P_{X^n}(x^n)c(x^n)\leq D\right\}.\label{eqn:sensingCost}
\end{align}
Distributions in $\overrightarrow{\mathcal{P}}_D^{(n)}$ satisfy the causality constraints of the system via the optimal estimator $g^*(X^n,Z^n)$ in the sensing cost definition. 

This section concludes by characterizing the open loop capacity region in Theorem \ref{thm:generalCapacity}, which makes use of the results of Lemma \ref{lem:optimalEstimator}.
\begin{theorem}\label{thm:generalCapacity}
    Given the model definition in Figure \ref{fig:systemModel}, the optimal estimator from Lemma \ref{lem:optimalEstimator}, and an open loop encoder, the capacity-distortion tradeoff function of the state dependent memoryless channel $P_{YZ\vert SX}$ with Markov state $P_{S_i\vert S_{i-1}}$ is
    \begin{equation}
        C^{(\textnormal{open})}(D)=\lim_{n\to\infty}\max_{P_{X^n}\in\overrightarrow{\mathcal{P}}_D^{(n)}} \frac{1}{n}\sum_{i=1}^n I(X_i;Y_i\vert S_i).
    \end{equation}
\end{theorem}
\begin{proof}
    See Appendix \ref{apd:capacity}
\end{proof}

Theorem \ref{thm:generalCapacity} presents a fairly abstract formulation for the JCAS capacity region, which makes it difficult to draw specific conclusions or insights.
To illustrate the behavior of these results, we specialize the model to a more concrete example.

\section{Beam Pointing Example}\label{sec:beamPointingExample}
We now specialize the results of Section \ref{sec:mainResults} to the beam pointing problem shown in Figure \ref{fig:beamPointingProblemSetup}. 
A transmitter wishes to communicate with a fixed receiver while simultaneously tracking a separate mobile target.
The receiver and mobile target are modeled as separate entities to highlight the tradeoff between sensing and communication operations. 
The beam pointing problem specializes the general model of Section \ref{sec:mainResults} using the following assumptions.
\begin{enumerate}[align=left, label=\textbf{Assumption \arabic*},leftmargin=0.7em]
    \setlength{\itemindent}{0em}
    \item\label{itm:bpDef1} Each channel input consists of a pair $(X,\Gamma)$.
    $X$ carries the encoded information, while $\Gamma$ captures the beam pointing strategy chosen by the transmitter.
    Said differently, the choice of $X$ does not encode beam direction or beam width.
    The quality of the transmission from a beam pointing perspective is determined only by $\Gamma$.
    The channel distribution under this assumption becomes $P_{YZ|X\Gamma S}$.
    Realizations of $(X,\Gamma)$ are denoted $(x_i,\gamma_i)$.
    \item\label{itm:bpDef2} The receiver and target are separate entities.
    Furthermore, the channel can be factored as $P_{YZ|X\Gamma S}=P_{Y|X\Gamma }P_{Z|X\Gamma S}$. 
    $P_{Y|X\Gamma}$ represents the communication channel, which does not depend on the state $S$, only on the choice of $\Gamma$.
    $P_{Z|X\Gamma S}$ represents the measurement channel, which does depend on the state of the mobile target.
    \item\label{itm:bpDef3} The state evolves according to a linear Gauss-Markov model.
    \item\label{itm:bpDef4} Consistent with Theorem \ref{thm:generalCapacity}, we consider open loop coding strategies for the channel inputs.
    \item\label{itm:bpDef5} The transmitter possesses an initial state estimate such that $d_0(s_0,\hat{s}_0)<D$.
    The sensing objective is to then only \textit{track} the state of the target rather than \textit{acquire and track}.
\end{enumerate}
Note that this system is primarily intended to build intuition and consequently ignores some finer details of realistic beamforming design such as side-lobe behaviors, beam switching times, blockages/multipath, etc.

Among the possible beam pointing strategies, we choose to present two beam pointing strategies to illustrate and explore how the general results may be specialized: beam switching and multi-beam.
The beam switching strategy mimics systems capable of transmitting only one beam at a time.
The beam must either commit to communication operation or sensing operation, but not both.
Such a strategy results in a clear tradeoff between time spent transmitting to the receiver and time sensing. 
Systems capable of transmitting multiple beams simultaneously have the option to sense and communicate at the same time.
The tradeoff in a multi-beam approach results from limited system resources such as power, bandwidth, and spectrum constraints.

\vspace{-0.5em}
\begin{figure}[htb!]
    \centering
    \def\svgwidth{0.8\linewidth}
    {\footnotesize
\begingroup%
  \makeatletter%
  \providecommand\color[2][]{%
    \errmessage{(Inkscape) Color is used for the text in Inkscape, but the package 'color.sty' is not loaded}%
    \renewcommand\color[2][]{}%
  }%
  \providecommand\transparent[1]{%
    \errmessage{(Inkscape) Transparency is used (non-zero) for the text in Inkscape, but the package 'transparent.sty' is not loaded}%
    \renewcommand\transparent[1]{}%
  }%
  \providecommand\rotatebox[2]{#2}%
  \newcommand*\fsize{\dimexpr\f@size pt\relax}%
  \newcommand*\lineheight[1]{\fontsize{\fsize}{#1\fsize}\selectfont}%
  \ifx\svgwidth\undefined%
    \setlength{\unitlength}{205.18942804bp}%
    \ifx\svgscale\undefined%
      \relax%
    \else%
      \setlength{\unitlength}{\unitlength * \real{\svgscale}}%
    \fi%
  \else%
    \setlength{\unitlength}{\svgwidth}%
  \fi%
  \global\let\svgwidth\undefined%
  \global\let\svgscale\undefined%
  \makeatother%
  \begin{picture}(1,0.42342073)%
    \lineheight{1}%
    \setlength\tabcolsep{0pt}%
    \put(0,0){\includegraphics[width=\unitlength,page=1]{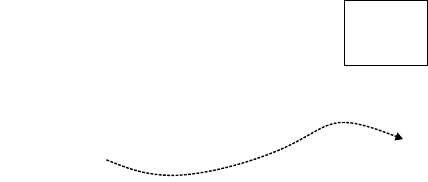}}%
    \put(0.85875955,0.35653688){\color[rgb]{0,0,0}\makebox(0,0)[lt]{\lineheight{1.25}\smash{\begin{tabular}[t]{c}Fixed\\Rx\end{tabular}}}}%
    \put(0,0){\includegraphics[width=\unitlength,page=2]{BeamPointingSetup.pdf}}%
    \put(0.05490049,0.35653688){\color[rgb]{0,0,0}\makebox(0,0)[lt]{\lineheight{1.25}\smash{\begin{tabular}[t]{c}Fixed\\Tx\end{tabular}}}}%
    \put(0,0){\includegraphics[width=\unitlength,page=3]{BeamPointingSetup.pdf}}%
    \put(0.59409545,0.07039434){\color[rgb]{0,0,0}\makebox(0,0)[t]{\lineheight{1.25}\smash{\begin{tabular}[t]{c}Mobile\\Target\end{tabular}}}}%
  \end{picture}%
\endgroup%
}
    \caption{Illustration of a beam pointing JCAS system with a fixed receiver and separate mobile target.}
    \label{fig:beamPointingProblemSetup}
\end{figure}
\vspace{-1.5em}

\subsection{Mobility Model and Kalman Filtering}
We assume that the state evolves according to a linear Gauss-Markov model of the form,
\begin{align}
    s_i&=As_{i-1}+w_{i-1}\\
    z_i&=Cs_i+v_i 
\end{align}
where $A\in\R^{m\times m}$, $w_{i-1}\in\R^m$ is the process noise distributed as $w_{i-1}\sim\mathcal{N}(0,Q)$, $C\in\R^{k\times m}$, $v_i\in\R^{k\times k}$ is the measurement noise distributed as $v_i\sim\mathcal{N}(0,\gamma_iR)$.
We assume that $(A,C)$ is detectable and that $(A,Q^{1/2})$ is controllable.

Consistent with \ref{itm:bpDef1}, $\gamma_i$ is a channel input that encodes the beam pointing strategy of the transmitter.
Operationally, $\gamma_i$ acts as a scalar gain to the measurement noise covariance matrix, capturing the fact that the transmitter's choice of beam strategy (e.g. beam width and beam direction) influences the quality of the state measurement and received message.
One can think of $\gamma_i$ as a description of how much power is being delivered to either the receiver or the target.
The exact behavior of $\gamma_i$ is determined by the chosen beam pointing strategy, defined in subsequent sections.

\begin{remark}
    We deliberately encode the tradeoff as varying measurement noise strength rather than varying the signal power stay consistent with prior work.
    Both approaches can be made equivalent as performance for both the sensing and communication operations depends on the signal to noise ratio (SNR) rather than power or noise alone.
\end{remark}

By \ref{itm:bpDef3}, the optimal Bayesian filter for the proposed beam pointing problem is a Kalman filter.
The distortion metric is the mean square error.
We use $\hat{s}_{j|k}$ to denote the optimal state estimate at time $j$ given available knowledge up to and including time $k$ for $j\geq k$.
Similarly, we denote the estimation error covariance at time $j$ given time $k$ as $P_{j|k}$.

The Kalman filter is initialized with an estimate $\hat{s}_0$ and covariance $P_0$ arising from \ref{itm:bpDef5}.
The prediction and update equations are well known with a slight change to the Kalman gain $K_i$ to include $\gamma_i$.
\begin{align}
    K_{i+1}&=P_{i+1|i}C^T(CP_{i+1|i}C^T+\gamma_{i+1}R)^{-1}.\label{eqn:IKF_kalmanGain}    
\end{align}
The single step covariance update is
\begin{equation}\label{eqn:discreteARE}
    P_{i+1}=AP_{i}A^T+Q-AP_{i}C^T(CP_{i}C^T+\gamma_iR)^{-1}CP_{i}A^T,
\end{equation}
where we use the simplifying notation $P_{i|i-1}\triangleq P_i$.
In classical Kalman filter analysis, a steady state covariance is found by finding a fixed point of (\ref{eqn:discreteARE}) as $i\to\infty$.
For the beam switching analysis, we treat $\gamma_i$ as a random variable, meaning (\ref{eqn:discreteARE}) is stochastic.
Consequently, we analyze the behavior of $E[P_i]$ as $i\to\infty$.

\subsection{Beam Switching Strategy}\label{sec:beamSwitchingStrategy}
In the beam switching strategy, the system can transmit only one beam and must switch between sensing and communication operations.
To model this behavior, the input sequence $\{\gamma_i\}$ is defined as a sequence of random variables such that
\begin{equation}\label{eqn:beamSwitchingGamma}
    \gamma_i=
    \begin{cases}
        1 & \text{ with probability } \lambda\\
        \sigma & \text{ with probability } (1-\lambda)
    \end{cases}
\end{equation}
where $\sigma$ is allowed to tend toward infinity.
Operationally, this means that the system is sensing with probability $\lambda$ and communicating with probability $1-\lambda$.

\begin{remark}\label{rmk:openLoopGamma}
    For $\gamma_i$ defined by (\ref{eqn:beamSwitchingGamma}), the sensing operation becomes an intermittent Kalman Filter, which is studied in further depth by \cite{sinopoli_kalman_2004}.
    A feature of intermittent Kalman filters is that as $\sigma$ tends to infinity, equation (\ref{eqn:discreteARE}) reduces to $P_{i+1}=AP_{i}A^T+Q$, implying that the filter is updating in open loop.
\end{remark}

\begin{lemma}\label{lem:bsExpectedARE}
    Under the beam switching strategy for $\gamma_i$, the expected value $E[P_i]$ is
    \begin{equation}
        E[P_i] = E[\Gamma_{\textnormal{bs}}(P_{i-1},\lambda)]
    \end{equation}
    where
    \begin{equation}\label{eqn:intermittantARE}
        \Gamma_{\textnormal{bs}}(P,\lambda)=APA^T+Q-\lambda APC^T(CPC^T+R)^{-1}CPA^T.
    \end{equation}
\end{lemma}
\begin{proof}
    See Appendix \ref{apd:bsExpectedARE}
\end{proof}

\begin{lemma}\label{lem:boundingExpectedCovariance}
    The cost constrained set $\overrightarrow{\mathcal{P}}_D^{(n)}$ defined in equation (\ref{eqn:sensingCost}) is
    \begin{equation}
        \overrightarrow{\mathcal{P}}_D^{(n)}=
        \left\{ P_{X^n\Gamma^n}|\textnormal{tr}(E[P_n])\leq D\right\}
    \end{equation}
    and satisfies
    \begin{align}\label{eqn:costConstraintInclusions}
        \left\{ P_{X^n\Gamma^n}|V_n\leq D\right\}
        \subseteq\overrightarrow{\mathcal{P}}_D^{(n)}\subseteq
        \left\{ P_{X^n\Gamma^n}|S_n\leq D\right\}
    \end{align}
    where $\lim_{n\to\infty}S_n=\bar{S}$ and $\lim_{n\to\infty}V_n=\bar{V}$ are solutions to the algebraic equations $\bar{S}=(1-\lambda)A\bar{S}A^T+Q$ and $\bar{V}=\Gamma_{\textnormal{bs}}(\bar{V},\lambda)$, respectively.
    Furthermore, 
    \begin{equation}
        \bar{S}\leq\lim_{n\to\infty}E[P_n]\leq\bar{V}.
    \end{equation}
\end{lemma}
\begin{proof}
    See Appendix \ref{apn:boundingExpectedCovariance}
\end{proof}

The only parameter that the encoder controls to affect $\text{tr}(E[P_n])$ is the probability $\lambda$ of collecting a measurement.
It would be analytically convenient to write the cost constraint in terms of the feasible $\lambda$.
However, it has not been proved that there generally exists some $\hat{\lambda}$ such that
$\lambda\geq\hat{\lambda}\iff\text{tr}(E[P_n])\leq D$ (see \cite{sinopoli_kalman_2004} Theorem 2).
In contrast, by continuity and monotonicity of the equations for $\bar{S}$ and $\bar{V}$, 
\begin{align}
    &\exists\lambda_S\text{ s.t. }\lambda\geq\lambda_S\iff\bar{S}\leq D\\
    &\exists\lambda_V\text{ s.t. }\lambda\geq\lambda_V\iff\bar{V}\leq D
\end{align}
where $\lambda_S\leq\lambda_V$.
Since $\bar{S}$ and $\bar{V}$ bound $\lim_{n\to\infty}\text{tr}(E[P_n])$, we can consider an inner and outer cost constrained set in terms of $\bar{S}$, $\bar{V}$ and feasible $\lambda$.

Let the outer cost constrained set $\mathcal{P}_{\Lambda_S}$ and inner cost constrained set $\mathcal{P}_{\Lambda_V}$ be described by
\begin{equation}
    \mathcal{P}_{\Lambda_S}(D)=\left(\left\{ \lambda|\bar{S}\leq D\right\}\right),
    \mathcal{P}_{\Lambda_V}(D)=\left(\left\{ \lambda|\bar{V}\leq D\right\}\right)
\end{equation}
where $\mathcal{P}_{\Lambda_V}\subseteq\mathcal{P}_{\Lambda_S}$.
These sets describe what can intuitively be thought of as a switching cost.
The following theorem provides inner and outer bounds for the JCAS rate distortion function for the beam switching strategy.
These bounds are not necessarily tight since $\lambda$ does not directly tune the value of $\lim_{n\to\infty}\text{tr}(E[P_n])$.
On the other hand, $\lambda$ does directly tune the values of $\bar{S}$ and $\bar{V}$, which then bound the expected error.

\begin{theorem}\label{thm:beamSwitchingCapacity}
    Given the beam switching strategy, the rate-distortion tradeoff function for the open loop beam pointing JCAS system is
    \begin{equation}
        \max_{
        \begin{aligned}
            {\scriptstyle P_{X}\in\mathcal{P}_X}\\[-0.5em]
            {\scriptstyle \lambda\in\mathcal{P}_{\Lambda_V}}
        \end{aligned}
        }(1-\lambda)I(X;Y)
        \leq C_{bs}(D)\leq
        \max_{
        \begin{aligned}
            {\scriptstyle P_{X}\in\mathcal{P}_X}\\[-0.5em]
            {\scriptstyle \lambda\in\mathcal{P}_{\Lambda_S}}
        \end{aligned}
        }(1-\lambda)I(X;Y).
    \end{equation}
\end{theorem}

\begin{proof}
    See Appendix \ref{apd:beamSwitchingCapacity}.
\end{proof}

\subsection{Multi-Beam Strategy}
In a multi-beam strategy, the system can beamform in multiple directions at once, allowing for simultaneous sensing and communication.
Each beam shares system resources, which results in a natural tradeoff between measurement quality and communication rate.

We model a simple open loop strategy in which the power allocated to each beam remains constant over the duration of the codeword.
This is modeled by letting $\gamma_i=\gamma_0$ for $i=1,2,\dots,n$.
$\gamma_0$ is allowed to take values in $[1,\infty)$ where $\gamma_0=1$ implies ``all sensing'' and $\gamma_0=\infty$ implies ``all communicating''.

Let $\Gamma_{\textnormal{mb}}(P,\gamma)$ be the multi-beam analogue of the modified ARE from (\ref{eqn:intermittantARE}).
\begin{equation}\label{eqn:multiBeamARE}
    \Gamma_{\textnormal{mb}}(P,\gamma)=APA^T+Q-APC^T(CPC^T+\gamma R)^{-1}CPA^T
\end{equation}
Since $\gamma_0$ is assumed constant, the steady state covariance is no longer random, and the steady state error of the estimator can be found as
\begin{equation}\label{eqn:multiBeamSteadyStateError}
    \text{tr}(P)=\text{tr}(\Gamma_{\textnormal{mb}}(P,\gamma_0)).
\end{equation}

Given (\ref{eqn:multiBeamSteadyStateError}), the estimation cost constrained set is no longer determined by the information bearing codewords $X^n$.
Rather, satisfaction of the distortion constraint is only determined by the $\gamma_0$ such that $\text{tr}(\Gamma_{mb}(P,\gamma_0))$ remains below the allowed distortion $D$.

\begin{theorem}\label{thm:multiBeamCapacity}
    Let 
    \begin{equation}\label{eqn:mbGammaSet}
        \mathcal{G}(D)=\{\gamma:\textnormal{tr}(\Gamma_{\textnormal{mb}}(P,\gamma_0))\leq D\}.
    \end{equation}
    Under the multi-beam strategy, the rate-distortion tradeoff function for the beam pointing JCAS system is
    \begin{equation}
        C_{\textnormal{mb}}(D) =\max_{
            {\scriptstyle P_{X}\in\mathcal{P}_X,\gamma_0\in\mathcal{G}(D)}
        }I(X;Y|\Gamma=\gamma_0).
    \end{equation}
\end{theorem}
\begin{proof}
    See Appendix \ref{apd:multiBeamCapacity}
\end{proof}

\subsection{Numerical Results}
We illustrate the rate distortion regions of the beam pointing strategies using two different state models.
Consider the two following scalar systems. 
\begin{align*}
    \textbf{Unstable system: }&A = -1.15, Q=0.2, C = 1, R=1.5\\
    \textbf{Stable system: }&A = -0.95, Q=0.2, C = 1, R=1.5
\end{align*}

These two systems differ only in $A$, where the first system is unstable, and the second system is stable.

For unstable systems, \cite{sinopoli_kalman_2004} shows that there exists a critical $\lambda>0$, below which the mean square error of an intermittent Kalman filter is no longer guaranteed to converge for all initial values and sequences $\{\gamma_i\}_i$.
This is illustrated in Figure \ref{fig:noiselessChannel}, which presents the bounds of a noiseless communication channel using the beam-switching strategy.

Assuming a noiseless discrete communication channel $P_{Y|X\Gamma}$, the expression in Theorem \ref{thm:beamSwitchingCapacity} becomes a maximization over $(1-\lambda)$ given some distortion constraint.
As Figure \ref{fig:noiselessChannel} shows, no finite allowed distortion can yield a communication rate of $1$ for the unstable system.
Alternatively, it can be said that there exists some communication rate above which there is no longer any guarantee of a finite distortion.
Thus, for the unstable system, some minimum amount of time must be spent sensing for feasible joint operation.
On the other hand, the total lack of state measurements ($\lambda=0$) still results in a finite distortion for the stable system.

We next assume a Gaussian channel and compare the beam switching and multi-beam strategies.
For the multi-beam illustration, the communication SNR and measurement noise are parameterized by $\gamma_0$ to sweep between ``all sensing'' ($\gamma_0=1$) and ``all communication'' ($\gamma_0=\infty$) modes.
The result is a power sharing scheme with the power divided between the two operations.
While such a simulation simplifies the complexities of a full multi-beam system, we believe the results suggest trends that arise from strategic resource allocation.
Given a constant $\gamma_0$, fixed power $P_0$, communication noise $N_0$, and previously defined sensing dynamics, the channel rate and sensing distortion for the numerical simulation are given by,
\begin{equation}
    C_{\textnormal{mb}} = \frac{1}{2}\log(1+\frac{\gamma_0-1}{\gamma_0}\frac{P_0}{N_0}),\hspace{1em}
    D_{\textnormal{mb}} = \text{tr}(\Gamma_{\textnormal{mb}}(P,\gamma_0))
\end{equation}

For the beam switching strategy and Gaussian channel, the rate and distortion equations are parameterized by $\lambda$ and are given by,
\begin{equation}
    C_{\textnormal{bs}} = (1-\lambda)\frac{1}{2}\log(1+\frac{P_0}{N_0}),\hspace{1em}
    \text{tr}(\bar{S}) \leq D_{\textnormal{bs}} \leq \text{tr}(\bar{V})
\end{equation}

\begin{figure}[htb!]
    \centering
    \includegraphics[width=0.49\linewidth]{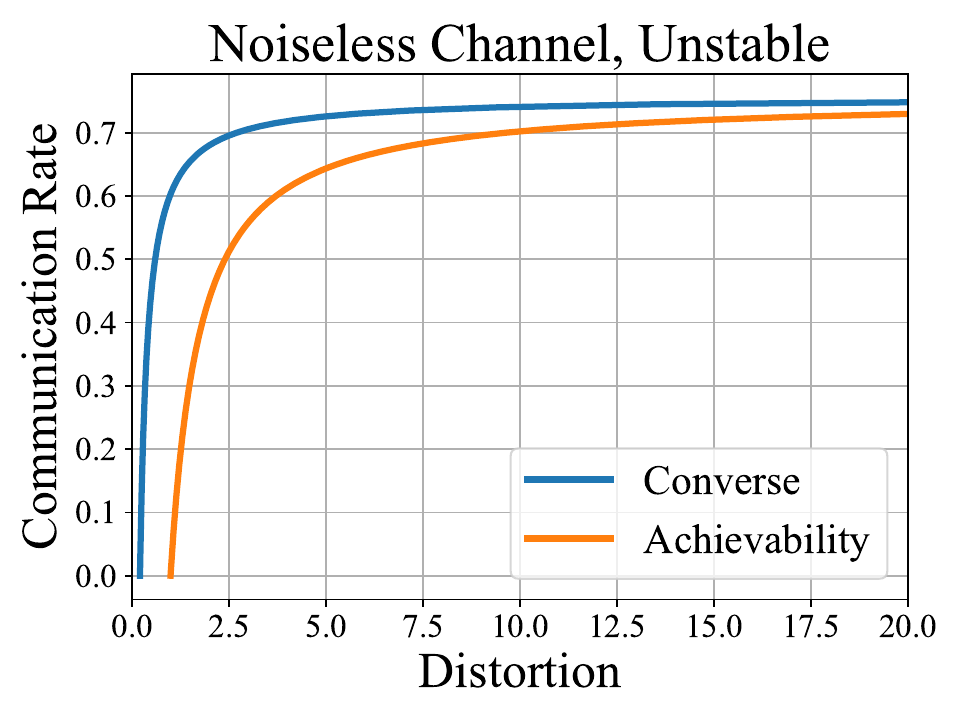}
    \includegraphics[width=0.49\linewidth]{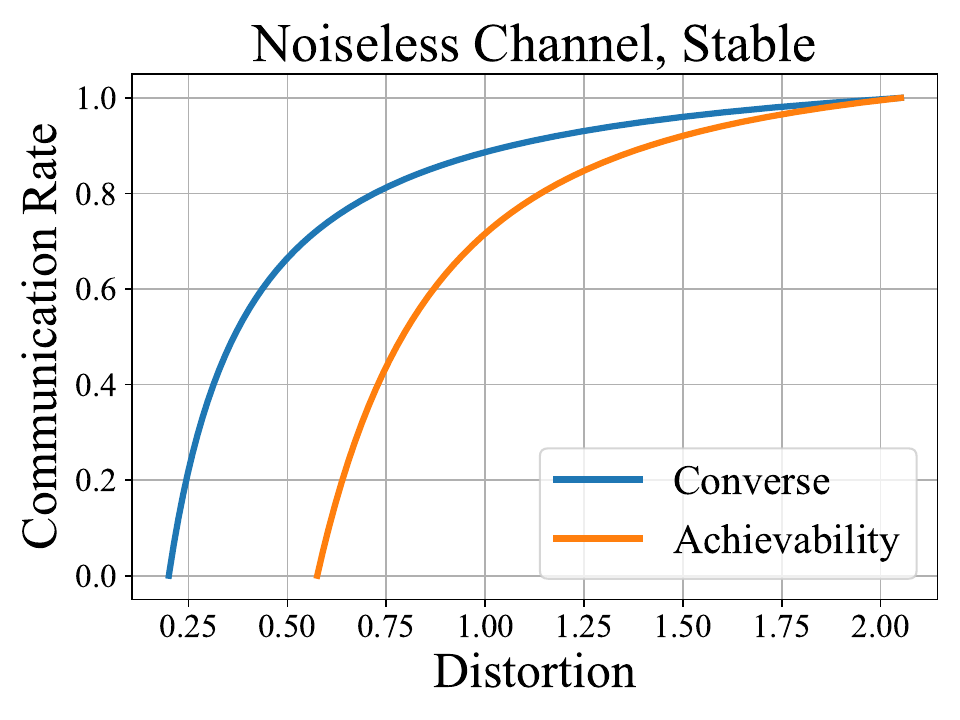}
    \caption{Rate-distortion regions for the unstable (left) and stable (right) system for a noiseless channel.}
    \label{fig:noiselessChannel}
\end{figure}
The rate distortion region for the comparison is computed using a Gaussian channel with an SNR of approximately 1.75 dB.
Results are given in Figure \ref{fig:mbStrategyComparison}.

\begin{figure}[h]
    \centering
    \includegraphics[width=0.49\linewidth]{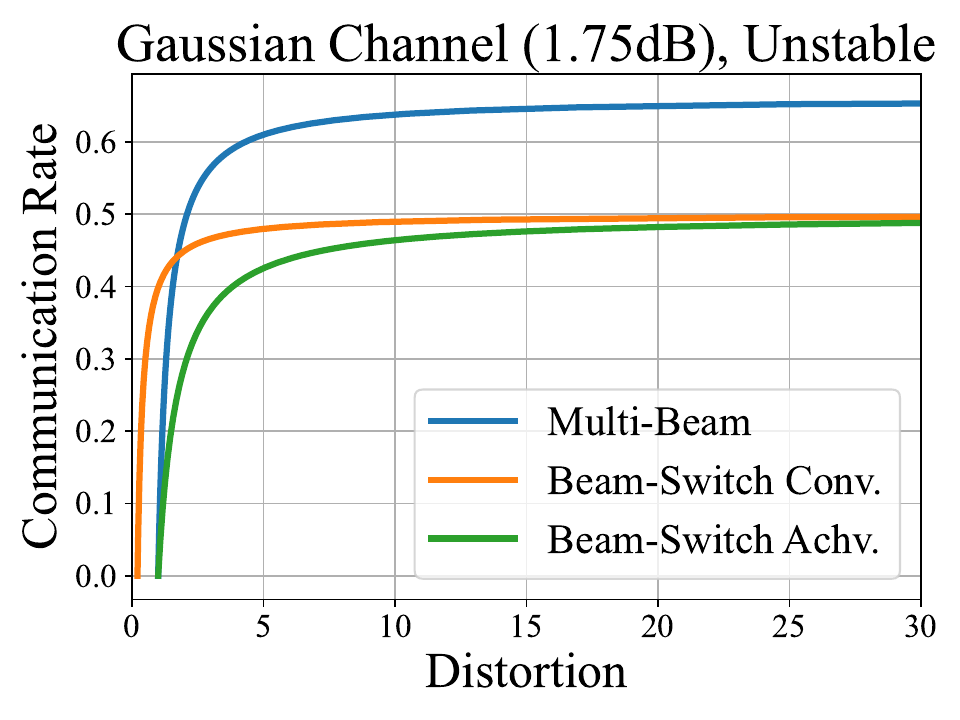}
    \includegraphics[width=0.49\linewidth]{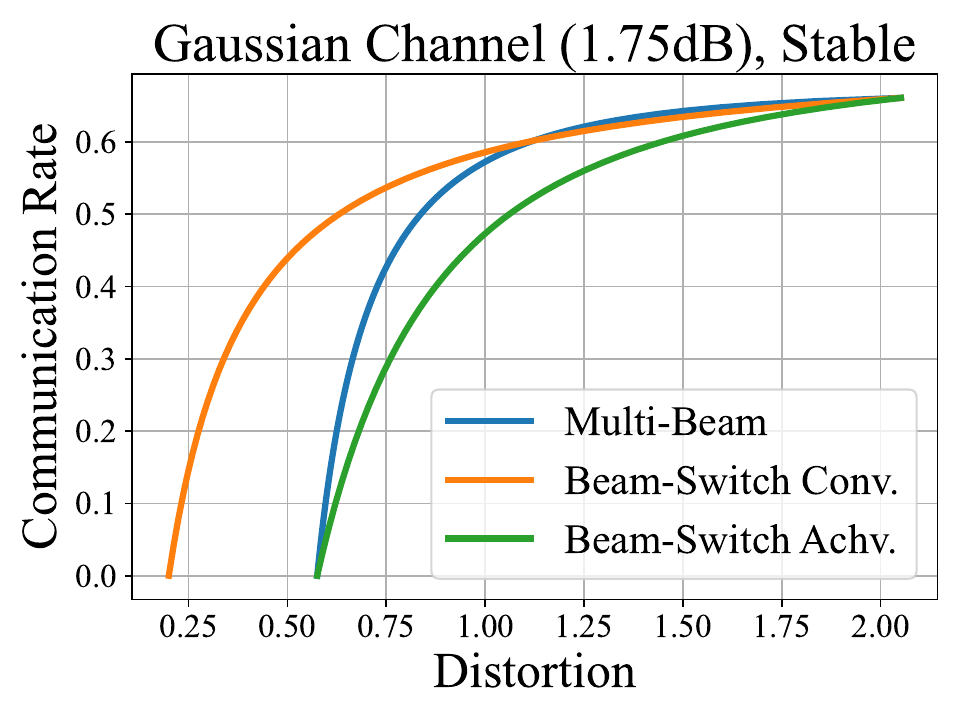}\\
    \includegraphics[width=0.49\linewidth]{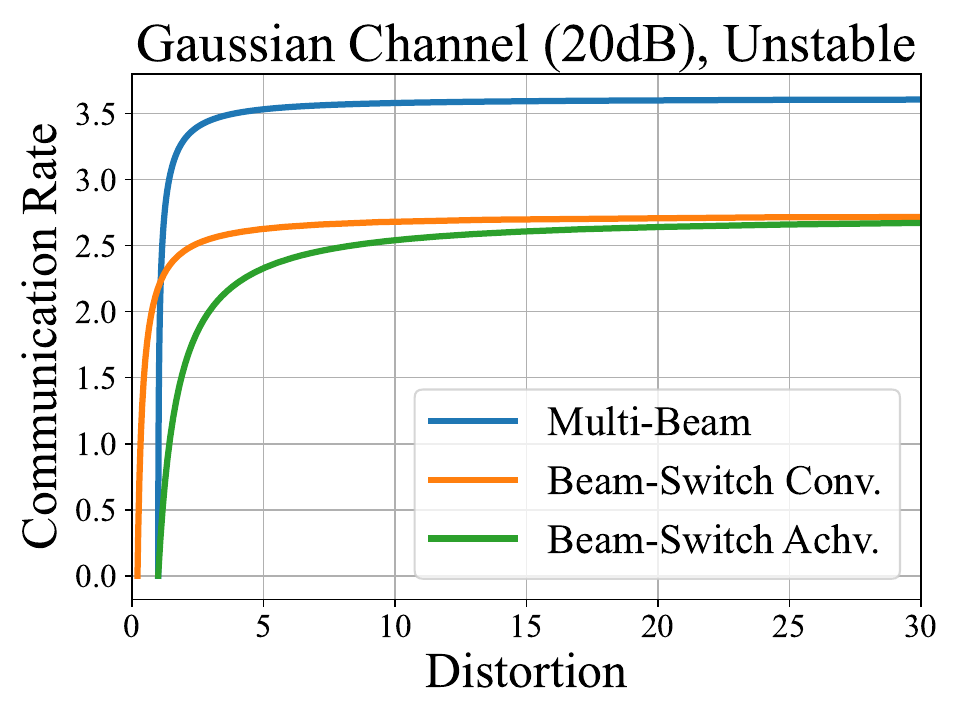}
    \includegraphics[width=0.49\linewidth]{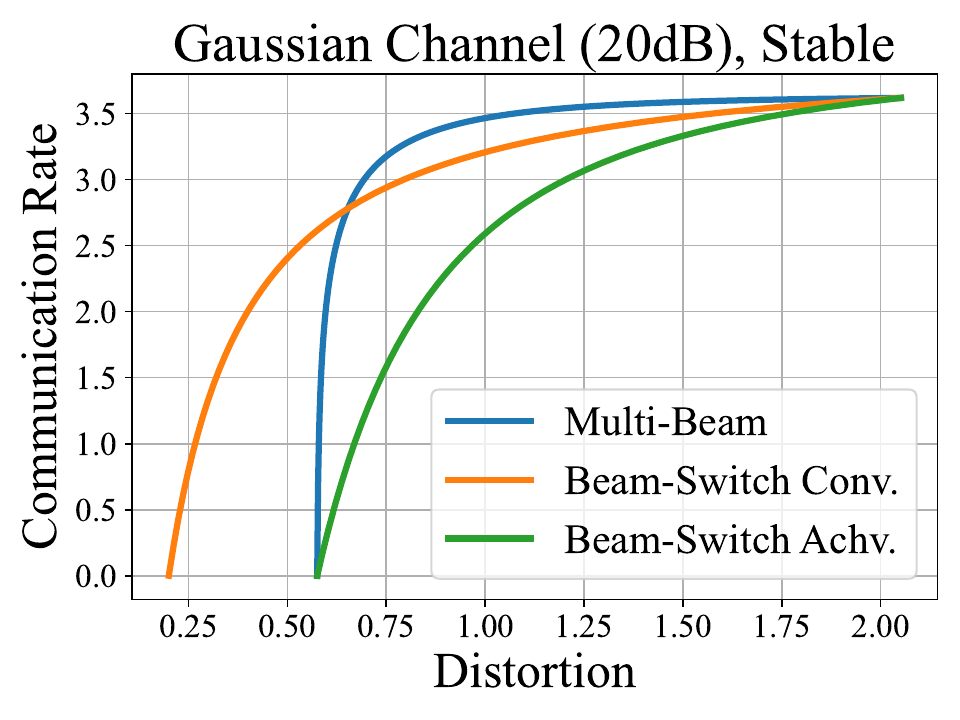}
    \caption{Comparison of strategies over Gaussian channel with SNR of 1.75 dB (top row) and SNR of 20 dB (bottom row) for the unstable (left) and stable (right) system.}
    \label{fig:mbStrategyComparison}
\end{figure}

For both the stable and unstable systems, the multi-beam strategy outperforms the beam-switching achievability bound.
For high communication rates, the multi-beam strategy also outperforms the beam-switching converse bound.
Furthermore, for the unstable system, the rate distortion region of the multi-beam strategy significantly outperforms the beam switching strategy at high communication rates.
This is because of the beam-switching strategy's tendency toward unbounded error for small probability of sensing, $\lambda$.
The multi-beam strategy does exhibit unbounded error for large enough sensing gain $\gamma_0$ but at higher communication rates.

Figure \ref{fig:mbStrategyComparison} also shows the same rate distortion curves comparing the two strategies using an SNR of approximately 20 dB. 
The performance gap between the two strategies widens at high SNR, but the same general features exist.

\bibliographystyle{IEEEtran}
\bibliography{citations}

\appendices

\section{Proof of Lemma \ref{lem:optimalEstimator} (Optimal estimator)}\label{apd:optimalEstimator}
By definition of the per-block distortion in (\ref{eqn:blockDistortionDefinition}), minimizing the average per-block distortion implies minimizing the average per-letter distortion in (\ref{eqn:letterDistortionDefinition}) for each time step $i\in[0;n]$.
As such, we minimize $E[d_i(S_i,\hat{S}_i)]$.
\begin{align}
    E[d_i(S_i,\hat{S}_i)]
    &=E_{X^i,Z^i}[d_i(S_i,\hat{S}_i)|X^i,Z^i]\label{eqn:A1_1}\\
    &=\sum_{x^i,z^i}P_{X^iZ^i}(x^i,z^i)\int_{\mathcal{S},\hat{\mathcal{S}}}P_{S_i,\hat{S}_i\vert X^i,Z^i}(s_i,\hat{s}_i\vert x^i,z^i)d_i(S_i,\hat{S}_i)ds_id\hat{s}_i\\
    &=\sum_{x^i,z^i}P_{X^iZ^i}(x^i,z^i)\int_{\hat{\mathcal{S}}}P_{\hat{S}_i\vert X^i,Z^i}(\hat{s}_i\vert x^i,z^i)\int_{\mathcal{S}}P_{S_i\vert X^i,Z^i}(s_i\vert x^i,z^i)d_i(S_i,\hat{S}_i)ds_id\hat{s}_i\\
    &\geq\sum_{x^i,z^i}P_{X^iZ^i}(x^i,z^i)\min_{\hat{s}_i}\int_{\mathcal{S}}P_{S_i\vert X^i,Z^i}(s_i\vert x^i,z^i)d_i(s_i,\hat{s}_i)ds_i\label{eqn:A1_2}\\
    &=E[d_i(S_i,\hat{S}^*_i)]
\end{align}
where the expectation in (\ref{eqn:A1_1}) is taken over $(X^i,Z^i)$ rather than $(X^n,Z^n)$ to conform to the causality constraint and where $\hat{S}^*_i=g_i^*(X^i,Z^i)$ denotes the optimal symbol estimate that minimizes (\ref{eqn:A1_2}) as,
\begin{equation}
    g_i^*(X^i,Z^i)=\argmin_{\hat{s}}\int_{\mathcal{S}}P_{S_i\vert X^i,Z^i}(s_i\vert x^i,z^i)d_i(s_i,\hat{s})ds_i.
\end{equation}
By forming an optimal estimate at each time step $i$, we form the optimal estimator of the whole sequence as
\begin{equation}
    g^*(X^n,Z^n)=\left(g_1^*(X^1,Z^1),g_2^*(X^2,Z^2),\dots,g_n^*(X^n,Z^n)\right).
\end{equation}

\section{Proof of Theorem \ref{thm:generalCapacity} (General capacity)}\label{apd:capacity}
\subsection{Converse Proof}
    Assume $(R,D)$ is achievable.
    We then prove that $R\leq C^{(\textnormal{open})}(D)$
    Following a standard approach, we obtain,
    \begin{align}
        nR&=H(W)\label{eqn:A2_1}\\
        &\leq I(W;Y^n,S^n)+H(W\vert\hat{W})\label{eqn:A2_2}\\
        &\leq I(W;Y^n\vert S^n) +1+P_c^{(n)}\log\vert\mathcal{M}\vert\label{eqn:A2_2a}\\
        &\leq I(X^n;Y^n\vert S^n)+\varepsilon'\label{eqn:A2_3}
    \end{align}
    where $\varepsilon'=1+P_c^{(n)}\log\vert\mathcal{M}\vert$.
    The mutual information can be single letterized as
    \begin{align}
        I(X^n;Y^n\vert S^n)&=\sum_{i=1}^n I(X^n;Y_i|Y^{i-1}S^n)\label{A2_9}\\
        &=\sum_{i=1}^n H(Y_i|Y^{i-1}S^n)-H(Y_i|X^nY^{i-1}S^n)\label{A2_10}\\
        &\leq\sum_{i=1}^n H(Y_i|S_i)-H(Y_i|X^nY^{i-1}S^n)\label{A2_11}\\
        &=\sum_{i=1}^n H(Y_i|S_i)-H(Y_i|X_iS_i)\label{A2_12}\\
        &=\sum_{i=1}^n I(X_i;Y_i|S_i)\label{A2_13}
    \end{align}
    where (\ref{A2_11}) holds because conditioning cannot increase entropy and (\ref{A2_12}) follows from the Markov chain $Y_i-(X_i,S_i)-(X^{i-1},S^{i-1},Y^{i-1},X_{i+1}^n,S_{i+1}^n)$ 
    
    Therefore,
    \begin{align}
        R&\leq\frac{1}{n}\sum_{i=1}^n I(X_i;Y_i|S_i)+\frac{1}{n}\varepsilon'\label{eqn:A2_4}\\ 
        &\leq\max_{P_{X^n}\in\overrightarrow{\mathcal{P}}_D^{(n)}}\frac{1}{n}\sum_{i=1}^n I(X_i;Y_i|S_i)+\frac{1}{n}\varepsilon'\label{eqn:A2_5}
    \end{align}
    where (\ref{eqn:A2_5}) is a maximization over the cost constrained set $\overrightarrow{\mathcal{P}}_D^{(n)}$ that arose as a consequence of Lemma \ref{lem:optimalEstimator}.
    Taking the limit as $n\to\infty$ and $\varepsilon'\to0$,
    \begin{equation}
        R \leq \lim_{n\to\infty} \max_{P_{X^n}\in\overrightarrow{\mathcal{P}}_D^{(n)}} \frac{1}{n} \sum_{i=1}^n I(X_i;Y_i|S_i)
        \triangleq C^{(\textnormal{open})}(D)\label{eqn:A2_6}.
    \end{equation}

\subsection{Achievability Proof}
    We need to prove that for any $R<C(D)$, $R,D$ is achievable according to the criteria from Definition $\ref{def:achievability}$.
    Consider an open loop coding strategy. 
    Fix $P_X$ and $g^*(X^n,Z^n)$ that achieve $R=C(\frac{D}{1+\delta})$ for some $\delta>0$.

    We use a random coding approach to show achievability.
    For $\alpha>0$, let
    \begin{equation}
        \mathcal{A}_\alpha\triangleq\left\{(x^n,y^n,s^n)\in\mathcal{X}^n\times\mathcal{Y}^n\times\mathcal{S}^n:\log\left( \frac{P_{Y^n|X^nS^n}(y^n|x^n,s^n}{P_{Y^n|S^n}(y^n|s^n)}\right) \right\}.
    \end{equation}
    Assume also that $|\mathcal{Y}|<\infty$ and that $\{S_i\}_{i\geq1}$ forms a homogeneous Markov chain.

    \begin{enumerate}
        \item \textit{Codebook generation:}
            Randomly generate $2^{nR}$ sequences $\{x^n(m)\}_{m=1}^{2^{nR}}$ where $m\in\{1,\dots,2^{nR}\}$ and $x^n\sim\prod_{i=1}^{n}p_X(x_i)$.
            These sequences constitute the codebook $\mathcal{C}$.
            Reveal the codebook to the encoder and decoder.
        \item \textit{Encoding:}
            To send the message $m\in\mathcal{M}$, the encoder transmits $x^n(m)$.
        \item \textit{Decoding:}
            After observing the sequences $Y^n=y^n$ and $S^n=s^n$, the decoder searches for a unique message $\hat{m}$ such that $(x^n(\hat{m}),y^n,s^n)\in\mathcal{A}_\alpha$.
            If such a unique $\hat{m}$ does not exist, then return a pre-defined arbitrary $m_0$.
            This defines $h(y^n,s^n)=\hat{m}$.
        \item \textit{Estimation:}
            Assuming the input sequence $x^n$ and measurement sequence $z^n$, the encoder computes estimates via the optimal estimator from Lemma \ref{lem:optimalEstimator}.
            \begin{equation}
                \hat{s}^n=g^*(x^n,z^n)=\left(g_1^*(x^1,z^1),g_2^*(x^2,z^2),\dots,g_n^*(x^n,z^n)\right)
            \end{equation}
        \item \textit{Analysis of Probability of Error:}
            Using threshold decoding, the expected probability of error over all length $n$ codebooks $C_n$ is known to be
            \begin{equation}
                E[P_e(C_n)]\leq P_{P_{X^n}P_{S^n}P_{Y^n|X^nS^n}}((X^n,Y^n,S^n)\not\in\mathcal{A}_\alpha^n)+M2^{-\alpha}.
            \end{equation}
            In the most general sense, any rate $R$ is achievable if
            \begin{equation}
                R<\sup_{\{P_{X^n}\}_{n\geq1}}\text{p-}\liminf_{n\to\infty}\frac{1}{n}\log\frac{P_{Y^n|X^nS^n}}{P_{Y^n|S^n}}
            \end{equation}
            where
            \begin{equation}
                \text{p-}\liminf_{n\to\infty}X_n=\sup_{\beta>0}\{\beta:\lim_{n\to\infty}P_{X_n}(X_n\leq\beta)=0\}.
            \end{equation}
            Given the memoryless channel, open loop encoding, and Markov state, we can perform the following factorizations.
            \begin{align}
                P_{Y^nX^nS^n}(y^n,x^n,s^n)
                &=\prod_{i=1}^n\left[P_{Y|XS}(y_i|x_i,s_i)P_X(x_i)P_{S_i|S_{i-1}}(s_i|s_{i-1})\right]\\
                P_{Y^nS^n}(y^n,s^n)
                &=\prod_{i=1}^n\left[P_{Y|S}(y_i|s_i)P_{S_i|S_{i-1}}(s_i|s_{i-1})\right]
            \end{align}
            Therefore,
            \begin{align}
                \log\frac{P_{Y^n|X^nS^n}}{P_{Y^n|S^n}}
                &=\log\frac{\prod_{i=1}^n P_{Y_i|X_iS_i}}{\prod_{i=1}^n P_{Y_i|S_i}}\\
                &=\sum_{i=1}^n\log\frac{P_{Y_i|X_iS_i}}{P_{Y_i|S_i}}
            \end{align}
            which is bounded almost surely because $|\mathcal{Y}|<\infty$.
            Assuming an information stable system with well-defined limit and a finite codebook, we can write $\sup_{\{P_{X^n}\}_{n\geq1}}\text{p-}\liminf_{n\to\infty}$ as $\lim_{n\to\infty}\max_{P_{X^n}}$.
            Thus, all rates $R$ are achievable that satisfy
            \begin{equation}
                R<\lim_{n\to\infty}\max_{P_{X^n}} \frac{1}{n}\sum_{i=1}^n\log\frac{P_{Y_i|X_iS_i}}{P_{Y_i|S_i}}.
            \end{equation}
            Since $\{S_i\}_{i\geq1}$ is a homogeneous Markov chain, we can use a Hoeffding type inequality to show that $\frac{1}{n}\sum_{i=1}^n\log\frac{P_{Y_i|X_iS_i}}{P_{Y_i|S_i}}$ concentrates around its expectation \cite{fan2021hoeffding}, given by
            \begin{align}
                E\left[\frac{1}{n}\sum_{i=1}^n\log\frac{P_{Y_i|X_iS_i}}{P_{Y_i|S_i}}\right]
                &=\frac{1}{n}\sum_{i=1}^nE\left[\log\frac{P_{Y_i|X_iS_i}}{P_{Y_i|S_i}}\right],\\
                &=\frac{1}{n}\sum_{i=1}^nI(X_i;Y_i|S_i).
            \end{align}
            Therefore, rates $R$ are achievable if
            \begin{equation}
                R<\lim_{n\to\infty}\max_{P_{X^n}}\frac{1}{n}\sum_{i=1}^nI(X_i;Y_i|S_i).
            \end{equation}
        \item \textit{Analysis of Expected Distortion:}
            Consider the following upper bound of the expected distortion (averaged over the random codebook).
            \begin{align}
                \Delta^{(n)} &= E[d_{0,n}(s_0^n,\hat{s}_0^n)]\\
                &= E_{C^n}[E[d_{0,n}(s_0^n,\hat{s}_0^n)|\hat{m}]\\
                &= E[d_{0,n}(s_0^n,\hat{s}_0^n)|\hat{m}=m]P(\hat{m}=m)
                + E[d_{0,n}(s_0^n,\hat{s}_0^n)|\hat{m}\neq m]P(\hat{m}\neq m)\\
                &\leq E[d_{0,n}(s_0^n,\hat{s}_0^n)|\hat{m}=m](1-P_e) + D_{\textnormal{max}}P_e
            \end{align}
            where $D_{\textnormal{max}}$ is the maximum distortion experienced by the estimator.
            $D_{\textnormal{max}}$ is assumed to be finite either by system definition or definition of the state space.
            Since $P_X$ achieves $C(\frac{D}{1+\delta})$, we have,
            \begin{align}
                \Delta^{(n)}\leq E[d_{0,n}(s_0^n,\hat{s}_0^n)|\hat{m}=m] + D_{\textnormal{max}}P_e \\
                \leq \frac{D}{1+\delta} + D_{\textnormal{max}}P_e.
            \end{align}
            As $\delta\to0$, we have
            \begin{equation}
                \Delta^{(n)}\leq D + \varepsilon'
            \end{equation}
            where $\varepsilon'\triangleq D_{\textnormal{max}}P_e$.
    \end{enumerate}
    We see that for rates $R<C(D)$, the conditions in Definition \ref{def:achievability} are satisfied.
    Thus, there exists a code such that $C(D)$ is achievable.

\section{Proof of Lemma \ref{lem:bsExpectedARE} (Beam switching expected error)}\label{apd:bsExpectedARE}
    Beginning with (\ref{eqn:discreteARE}) we use Remark \ref{rmk:openLoopGamma} to compute the following.
    \begin{align}
        P_i&= AP_{i-1}A^T+Q-AP_{i-1}C^T(CP_{i-1}C^T+\mathbbm{1}\{\gamma_i=1\}R +\mathbbm{1}\{\gamma_i=\sigma\}\sigma R)^{-1}CP_{i-1}A^T\\
        &=AP_{i-1}A^T+Q-\mathbbm{1}\{x_i\in\mathcal{X}_s\}AP_{i-1}C^T(CP_{i-1}C^T+R)^{-1}CP_{i-1}A^T
    \end{align}
    Taking the expectation of both sides and using the law of total expectation,
    \begin{align}
        E[P_i] &= E[E[AP_{i-1}A^T+Q-\mathbbm{1}\{x_i\in\mathcal{X}_s\}AP_{i-1}C^T(CP_{i-1}C^T+R)^{-1}CP_{i-1}A^T|P_{i-1}]]\\
        &= E[AP_{i-1}A^T+Q-\lambda AP_{i-1}C^T(CP_{i-1}C^T+R)^{-1}CP_{i-1}A^T]\\
        &= E[\Gamma_{bs}(P_{i-1},\lambda)]
    \end{align}
    where $\Gamma_{bs}(P,\lambda)$ is defined in (\ref{eqn:intermittantARE}).

\section{Proof of Lemma \ref{lem:boundingExpectedCovariance} (Bounding expected covariance)}\label{apn:boundingExpectedCovariance}
    Under a mean square error distortion, the average per-block distortion is defined as the expected value of the trace of the estimation error covariance.
    By linearity of the trace and expectation operators, 
    \begin{equation}
        \Delta^{(n)} = E\left[d_{0,n}(s_0^n,\hat{s}_0^n)\right]
        = E[\text{tr}(P_n)] = \text{tr}(E[P_n]).
    \end{equation}

    By Theorem 4 of \cite{sinopoli_kalman_2004},
    \begin{equation}\label{eqn:pnBounding}
        S_n\leq E[P_n]\leq V_n,\hspace{1em}\forall n
    \end{equation}
    where $\lim_{n\to\infty}S_n=\bar{S}$ and $\lim_{n\to\infty}V_n=\bar{V}$ are solutions to the algebraic equations $\bar{S}=(1-\lambda)A\bar{S}A^T+Q$ and $\bar{V}=\Gamma_{\textnormal{bs}}(\bar{V},\lambda)$, respectively.
    (\ref{eqn:costConstraintInclusions}) follows from (\ref{eqn:pnBounding}).

\section{Proof of Theorem \ref{thm:beamSwitchingCapacity} (Beam switching capacity)}\label{apd:beamSwitchingCapacity}
\subsection{Converse Proof}
    Using steps identical to the proof of Theorem \ref{thm:generalCapacity}, we arrive at
    \begin{equation}
        nR\leq\sum_{i=1}^n I(X_i;Y_i|S_i)+\varepsilon'
    \end{equation}
    where $\varepsilon'\triangleq 1+P_c^{(n)}\log\vert\mathcal{M}\vert$.
    Specializing to the beam pointing problem, we substitute $X_i\mapsto(X_i,\Gamma_i)$, then compute the following,
    \begin{align}
        nR&\leq\sum_{i=1}^n I(X_i\Gamma_i;Y_i|S_i)+\varepsilon'\\
        &\leq\sum_{i=1}^n I(X_i;Y_i|S_i\Gamma_i)+\varepsilon'\\
        &=\sum_{i=1}^n I(X_i;Y_i|S_i\Gamma_i=1)\P_{\Gamma}\{\gamma_i=1\}+
        I(X_i;Y_i|S_i\Gamma_i=\infty)\P_{\Gamma}\{\gamma_i=\infty\} + \varepsilon'\label{eqn:bsCapacityProof1}\\
        &=(1-\lambda)\sum_{i=1}^n I(X_i;Y_i|S_i) + \varepsilon'\label{eqn:bsCapacityProof2}\\
        &=(1-\lambda)\sum_{i=1}^n I(X_i;Y_i) + \varepsilon'\label{eqn:bsCapacityProof3}
    \end{align}
    where (\ref{eqn:bsCapacityProof1}) follows from the definition of conditional mutual information,
    (\ref{eqn:bsCapacityProof2}) is a result of $\gamma_i=1$ implying zero communication,
    and (\ref{eqn:bsCapacityProof3}) follows from \ref{itm:bpDef2}.
    
    Next, perform the following maximization using the outer cost constrained set $\mathcal{P}_{\Lambda_S}$.
    \begin{align}
        (1-\lambda)\sum_{i=1}^n I(X_i;Y_i) + \varepsilon'
        &\leq\max_{\lambda\in\mathcal{P}_{\Lambda_S}}(1-\lambda)\sum_{i=1}^n I(X_i;Y_i) + \varepsilon'\label{eqn:bsCapacityProof4}\\
        &\leq\max_{\lambda\in\mathcal{P}_{\Lambda_S}}(1-\lambda)n\max_{P_X\in\mathcal{P}_X} I(X;Y) + \varepsilon'\label{eqn:bsCapacityProof5}
    \end{align}

    Dividing by $n$ and taking the limit yields, 
    \begin{align}
        R&\leq\lim_{n\to\infty}\bigg(\max_{
        \begin{aligned}
            {\scriptstyle P_{X}\in\mathcal{P}_X}\\[-0.5em]
            {\scriptstyle \lambda\in\mathcal{P}_{\Lambda_S}}
        \end{aligned}
        }(1-\lambda)I(X;Y)+\frac{1}{n}\varepsilon'\bigg)\\
        &=\max_{
        \begin{aligned}
            {\scriptstyle P_{X}\in\mathcal{P}_X}\\[-0.5em]
            {\scriptstyle \lambda\in\mathcal{P}_{\Lambda_S}}
        \end{aligned}
        }(1-\lambda)I(X;Y).
    \end{align}
    
\subsection{Achievability Proof}
    Let
    \begin{equation}
        C_V(D) = \max_{
        \begin{aligned}
            {\scriptstyle P_{X}\in\mathcal{P}_X}\\[-0.5em]
            {\scriptstyle \lambda\in\mathcal{P}_{\Lambda_V}}
        \end{aligned}
        }(1-\lambda)I(X;Y).
    \end{equation}
    We prove the claim that for any $R<C_V(D)$, $R,D$ is achievable according to the criteria from Definition $\ref{def:achievability}$.
    
    Fix $P_{X\Gamma}$ and $g^*(X^n,Z^n)$ that achieve $R=C_V(\frac{D}{1+\delta})<C_V(D)$ for some $\delta>0$.
    \begin{enumerate}
        \item \textit{Codebook generation:}
            Randomly generate $2^{nR}$ sequences $\{x^n(m)\}_{m=1}^{2nR}$ where $x^n\sim\prod_{i=1}^{n}p_X(x)$.
            These sequences constitute the codebook $\mathcal{C}$.
            Reveal the codebook to the encoder and decoder.

            Under the beam switching model, the sequence $\gamma^n\sim\prod_{i=1}^np(\gamma_i)$ represents a symbol erasure, where $p(\gamma_i)$ is the distribution resulting in \ref{eqn:beamSwitchingGamma}.
            It follows that with probability $1-\lambda$, the symbol will pass through the communication channel.
            With probability $\lambda$, the symbol is erased because of a sensing operation.
        \item \textit{Encoding:}
            To send the message $m\in\mathcal{M}$, the encoder transmits $x^n(m)$.
        \item \textit{Decoding:}
            Let $\mathcal{A}_\varepsilon^{(n)}(P_{XY})$ be the set of jointly-typical sequences of $X$ and $Y$.
            After observing the sequence $Y^n=y^n$, the decoder searches for a message $\hat{m}$ such that
            \begin{equation}
                (x^n(\hat{m}),y^n)\in\mathcal{A}_\varepsilon^{(n)}(P_{XY})
            \end{equation}
        \item \textit{Estimation:}
            Assuming the input sequence $\gamma^n$ and measurement sequence $z^n$, the encoder computes estimates via the intermittent Kalman filter.
            \begin{equation}
                \hat{s}^n=(\hat{s}_{1|1}(\gamma_1,z_1),\hat{s}_{2|2}(\gamma_2,z_2)\cdots\hat{s}_{n|n}(\gamma_n,z_n))
            \end{equation}
        \item \textit{Analysis of Probability of Error:}
            By the symmetry of the code, we restrict our attention to $m=1$.
            We define the following communication error events.
            \begin{align}
                \mathcal{E}_1&=\{(X^n(1),Y^n) \not\in \mathcal{A}_\varepsilon^{(n)}\}\\
                \mathcal{E}_2&=\{(X^n(m'),Y^n) \in \mathcal{A}_\varepsilon^{(n)},m'\neq1\}
            \end{align}
            Note that the state sequence $S^n$ is not included in the error events because of the assumption that the channel model factors as $P_{YZ|XS}=P_{Y|X}P_{Z|XS}$.

            The probability of error is then
            \begin{equation}
                P_e^{(n)}=P(\mathcal{E}_1\cup\mathcal{E}_2)\leq P(\mathcal{E}_1)+P(\mathcal{E}_2).
            \end{equation}
            $P(\mathcal{E}_1)$ tends to $0$ as $n\to\infty$ by the law of large numbers.
            Using joint typicality arguments (keeping the probability of `erasure' in mind) $P(\mathcal{E}_2)$ tends to $0$ for rates $R<(1-\lambda)I(X;Y)$.
        \item \textit{Analysis of Expected Distortion:}
            From the achievability proof of Theorem \ref{thm:generalCapacity}, we have
            \begin{equation}
                \Delta^{(n)}\leq E[d_{0,n}(s_0^n,\hat{s}_0^n)|\hat{m}=1]+D_{\textnormal{max}}P_e.
            \end{equation}
            By construction of the beam pointing problem, $D_{\textnormal{max}}$ is finite.
            By Lemma \ref{lem:boundingExpectedCovariance},
            \begin{align}
                \Delta^{(n)}&\leq E[d_{0,n}(s_0^n,\hat{s}_0^n)|\hat{m}=1]+D_{\textnormal{max}}P_e \\
                &= \text{tr}(E[P_n])+\varepsilon' \\
                &\leq V_n+\varepsilon'
            \end{align}
            where $\varepsilon'\triangleq D_{\textnormal{max}}P_e$.
            Taking the limit as $n\to\infty$ yields
            \begin{align}
                \lim_{n\to\infty}\Delta^{(n)}&\leq\bar{V}\leq \frac{D}{1+\delta}
            \end{align}
            By continuity, $C_V(\frac{D}{1+\delta})$ approaches $C_V(D)$ as $\delta\to0$.
    \end{enumerate}
    Thus, there exists a code such that $C_V(D)$ is achievable.

\section{Proof of Theorem \ref{thm:multiBeamCapacity} (Multi-beam capacity)}\label{apd:multiBeamCapacity}
\subsection{Converse Proof}
    We have shown previously that
    \begin{equation}
        nR\leq\sum_{i=1}^n I(X_i;Y_i|S_i)+\varepsilon'
    \end{equation}
    where $\varepsilon'=1+P_c^{(n)}\log\vert\mathcal{M}\vert$.
    Specializing to the beam pointing problem, we substitute $X_i\mapsto(X_i,\Gamma_i)$, then compute the following.
    \begin{align}
        nR&\leq\sum_{i=1}^n I(X_i\Gamma_i;Y_i|S_i)+\varepsilon'\\
        &\leq\sum_{i=1}^n I(X_i;Y_i|S_i\Gamma_i)+\varepsilon'\\
        &=\sum_{i=1}^n I(X_i;Y_i|S_i\Gamma_i=\gamma_0) + \varepsilon'\label{eqn:mbCapacityProof1}\\
        &=\sum_{i=1}^n I(X_i;Y_i|\Gamma_i=\gamma_0) + \varepsilon'\label{eqn:mbCapacityProof2}
    \end{align}
    where (\ref{eqn:mbCapacityProof1}) follows from the multi-beam strategy and (\ref{eqn:mbCapacityProof2}) follows from \ref{itm:bpDef2}.

    We maximize over $X_i$ and $\gamma_0$ separately to incorporate the distortion constraint.
    \begin{align}
        \sum_{i=1}^n I(X_i;Y_i|\Gamma_i=\gamma_0) + \varepsilon'
        &=\sum_{i=1}^n \max_{\gamma_0\in\mathcal{G}(D)} I(X_i;Y_i|\Gamma_i=\gamma_0) + \varepsilon'\\
        &= n\max_{P_{X}\in\mathcal{P}_X}\max_{\gamma_0\in\mathcal{G}(D)} I(X;Y|\Gamma=\gamma_0) + \varepsilon'
    \end{align}

    Dividing by $n$ and taking the limit yields,
    \begin{align}
        R&\leq\lim_{n\to\infty}\bigg(\max_{
        \begin{aligned}
            &{\scriptstyle P_{X}\in\mathcal{P}_X}\\[-0.5em]
            &{\scriptstyle \gamma_0\in\mathcal{G}(D)}
        \end{aligned}}
        I(X;Y|\Gamma=\gamma_0) + \frac{1}{n}\varepsilon'\bigg)\\
        &=\max_{
        \begin{aligned}
            {\scriptstyle P_{X}\in\mathcal{P}_X}\\[-0.5em]
            {\scriptstyle \gamma_0\in\mathcal{G}(D)}
        \end{aligned}
        }I(X;Y|\Gamma=\gamma_0).
    \end{align}
    
\subsection{Achievability Proof}
    Fix $P_{X}$, $\gamma_0$ and $g^*(X^n,Z^n)$ that achieve $R=C_{mb}(\frac{D}{1+\delta})$ for some $\delta>0$.
    \begin{enumerate}
        \item \textit{Codebook generation:}
            Randomly generate $2^{nR}$ sequences $\{x^n(m)\}_{m=1}^{2nR}$ where $x^n\sim\prod_{i=1}^{n}p_X(x)$.
            These sequences constitute the codebook $\mathcal{C}$.
            Reveal the codebook to the encoder and decoder.

            Under the multi-beam model, $\gamma_0$ is a known gain that parameterizes the channel model.
        \item \textit{Encoding:}
            To send the message $m\in\mathcal{M}$, the encoder transmits $x^n(m)$.
        \item \textit{Decoding:}
            After observing the sequence $Y^n=y^n$, the decoder searches for a message $\hat{m}$ such that
            \begin{equation}
                (x^n(\hat{m}),y^n)\in\mathcal{A}_\varepsilon^{(n)}(P_{XY})
            \end{equation}
            where $\mathcal{A}_\varepsilon^{(n)}(P_{XY})$ is again the set of jointly typical inputs and outputs.
        \item \textit{Estimation:}
            Assuming the known gain $\gamma_0$ and measurement sequence $z^n$, the encoder computes estimates via the standard Kalman filter.
            \begin{equation}
                \hat{s}^n=(\hat{s}_{1|1}(z_1),\hat{s}_{2|2}(z_2)\cdots\hat{s}_{n|n}(z_n))
            \end{equation}
        \item \textit{Analysis of Probability of Error:}
            By the symmetry of the code, we restrict our attention to $m=1$.
            We define the following communication error events.
            \begin{align}
                \mathcal{E}_1&=\{(X^n(1),Y^n) \not\in \mathcal{A}_\varepsilon^{(n)}\}\\
                \mathcal{E}_2&=\{(X^n(m'),Y^n) \in \mathcal{A}_\varepsilon^{(n)},m'\neq1\}
            \end{align}
            Note that the state sequence $S^n$ is not included in the error events because of the assumption that the channel model factors as $P_{YZ|XS}=P_{Y|X}P_{Z|XS}$.

            The probability of error is then
            \begin{equation}
                P_e^{(n)}=P(\mathcal{E}_1\cup\mathcal{E}_2)\leq P(\mathcal{E}_1)+P(\mathcal{E}_2).
            \end{equation}
            $P(\mathcal{E}_1)$ tends to $0$ as $n\to\infty$ by the law of large numbers.
            Using joint typicality arguments $P(\mathcal{E}_2)$ tends to $0$ for rates $R<I(X;Y|\Gamma=\gamma_0)$.
        \item \textit{Analysis of Expected Distortion:}
            From the achievability proof of Theorem \ref{thm:generalCapacity}, we have
            \begin{equation}
                \Delta^{(n)}\leq E[d_{0,n}(s_0^n,\hat{s}_0^n)|\hat{m}=1]+D_{\textnormal{max}}P_e.
            \end{equation}
            By construction of the beam pointing problem, $D_{\textnormal{max}}$ is finite.
            Continuing, we have,
            \begin{align}
                \Delta^{(n)}&\leq E[d_{0,n}(s_0^n,\hat{s}_0^n)|\hat{m}=1]+D_{\textnormal{max}}P_e \\
                &= \text{tr}(P_n)+\varepsilon'
            \end{align}
            where $\varepsilon'=D_{\textnormal{max}}P_e$.
            Taking the limit as $n\to\infty$ yields
            \begin{align}
                \lim_{n\to\infty}\Delta^{(n)}&\leq\text{tr}(P)\leq \frac{D}{1+\delta}
            \end{align}
            where $P$ is the solution to the algebraic ricatti equation $\Gamma_{\textnormal{mb}}(P,\gamma_0)$.
            Letting $\delta$ tend to $0$ completes the proof.
    \end{enumerate}
    Therefore, there exists a code such that $C_{mb}(D)$ is achievable.

\end{document}